\documentclass[conference]{IEEEtran}
\IEEEoverridecommandlockouts
\usepackage[utf8x]{inputenc}
\usepackage{amsmath}
\usepackage{amssymb}
\usepackage{graphicx}
\usepackage{color}
\usepackage{amsthm}
\usepackage{bbm}
\usepackage{hyperref}
\usepackage[normalem]{ulem}
\newcommand{\pr}{\mathbb{P}}
\newcommand{\ex}{\mathbb{E}}

\newcommand\laweq{\mathrel{\overset{\makebox[0pt]{\mbox{\tiny law}}}{=}}}
\newtheorem{defi}{Definition}
\newtheorem{theo}{Theorem}
\newtheorem{prop}{Proposition}
\newtheorem{rem}{Remark}

\usepackage{xcolor}
\newcommand{\andi}[1]{\textcolor{purple}{#1}}

\begin{document}
\title{
Mitigation of liveness attacks in DAG-based ledgers
}

\author{
\IEEEauthorblockN{Darcy Camargo\IEEEauthorrefmark{3},  Andreas {Penzkofer}\IEEEauthorrefmark{1}, Sebastian {Müller}\IEEEauthorrefmark{2},  William Sanders\IEEEauthorrefmark{1}}
\IEEEauthorblockA{\IEEEauthorrefmark{1}IOTA Foundation, Berlin, Germany, Email: research@iota.org}
\IEEEauthorblockA{\IEEEauthorrefmark{2}Aix-Marseille Universit\'e, 
 CNRS, 
 I2M, UMR 7373, 13453 Marseille, France, Email: sebastian.muller@univ-amu.fr}
 \IEEEauthorblockA{\IEEEauthorrefmark{3}IOTA Foundation, Berlin, Germany, Email: darcy.camargo@iota.org}
}

\IEEEoverridecommandlockouts


\maketitle
\pagestyle{plain}

\begin{abstract}
The robust construction of the ledger data structure is an essential ingredient for the safe operation of a distributed ledger. While in traditional linear blockchain systems, permission to append to the structure is leader-based, in Directed Acyclic Graph-based ledgers, the writing access can be organised leaderless. However, this leaderless approach relies on fair treatment of non-referenced blocks, i.e. tips, by honest block issuers. 

We study the impact of a deviation from the standard tip selection by a subset of block issuers with the aim of halting the confirmation of honest blocks entirely. We provide models on this so-called \textit{orphanage} of blocks and validate these through open-sourced simulation studies. A critical threshold for the adversary issuance rate is shown to exist, above which the tip pool becomes unstable, while for values below the orphanage decrease exponentially. We study the robustness of the protocol with an expiration time on tips, also called garbage collection, and modification of the parent references per block. 
\end{abstract}


\section{Introduction}

Distributed Ledger Technology (DLT) must solve several core problems between participants in a network, called nodes, that do not necessarily trust each other.  Nodes must agree which transactions, from a set of transactions that is provided to them, should be added to a shared ledger. Since any pair of transactions may request conflicting operations -- a situation called double spending -- an agreement must be reached which transaction should be ultimately applied, a status we call \textit{confirmed}. All of the above must be handled in a byzantine environment, i.e.~in the presence of faulty nodes or malicious actors. 

To fulfil these requirements, in linear blockchain systems, such as Bitcoin, transactions are batched in blocks, which are appended to each other in a chain. A total order on the processing of transactions can be inferred through the chain, and consensus is reached by totally ordering the set of transactions. To guarantee security, the number of blocks per time interval is limited through a writing access control. The limitation can be probabilistic in time, as in Bitcoin through Proof-of-Work, or by round-based leader election, as in Proof-of-Stake systems. Since DLTs operate in a distributed setting, it is possible that simultaneously blocks are being created that attempt to extend the same endpoint, or \textit{tip}, of the chain. As a consequence, a fork is created in the blockchain. Additional rules, such as the longest chain rule in Bitcoin, must then be applied to identify the canonical chain. 

Increasingly, DLTs are also embracing Directed Acyclic Graph (DAG) data structures with their own appropriate consensus rules, writing access and block distribution. Nodes in such a system are enabled to issue blocks in parallel, with each block referencing several previous blocks. Consequently, the resulting graph infers a partial order to the blocks. Examples for such protocols are Aleph \cite{Gagol2019}, Sui \cite{sui2022} with Narwhal and Tusk \cite{Danezis2022}, 
and Iota \cite{muller2022}.

By the nature of the distributed setting of the protocol and the parallel writing access, as well as the possibility for faults or byzantine actors, conflicts may arise between issued transactions. 
Consensus on the accepted transactions can then be ensured by retrospective leadership election, e.g.~by using a common random coin \cite{Gagol2019, Danezis2022}, 
or by summation of votes \cite{muller2022,Avalanche19}. 
While the exact mechanism for the consensus  differs, the protocols have in common that the construction of the block DAG plays a crucial role in the confirmation of transactions and the performance of the consensus module. This construction is regulated through three aspects.

First, similar to the case of linear blockchain systems, to protect the resources in the network, a writing access mechanism is required that determines which nodes are permitted to add blocks to the DAG. 
Writing access can be granted in rounds in a permissioned system \cite{Gagol2019}, or through a separate peer-to-peer-based mechanism \cite{CongestionControl}. 

Second, each node must locally maintain a set of blocks that form the tips of the DAG, forming the so-called \textit{tip pool}, i.e. the set of blocks that the selection algorithm will choose from to approve when a new block is issued. This tip pool can also be understood as a distributed mempool that gives rise to a causally structured record of the communication between nodes. The reference scheme between blocks enables an efficient, reliable broadcast protocol, which can be a necessary prerequisite to the consensus protocol, e.g.~\cite{Danezis2022}.

Third, DLTs typically require a mechanism for nodes to prune old transactions.  One method of doing this is not considering blocks with a parent, i.e.~a directly referenced block, that is ``too far'' into the past.  Implementing this rule in a DAG-based protocol means that tips can \textit{expire} after a certain age.  Such blocks will receive no approvers, thus becoming \textit{orphaned}. For comparison, in \cite{Danezis2022}, this mechanism is described as garbage collection.


\subsection{Results}

In this work, we analyse the resilience of the protocol under attacks designed to (1) inflate the number of tips as a type of resource attack and (2) increase the number of tips which expire. Our analysis shows that these attacks are intertwined.
 
Maintaining a healthy tip pool buffer is necessary for a node to operate in any DAG-based DLT setting. In fact, a malicious actor intending to harm the protocol's consensus process could do so by inflating nodes' tip pools until they are forced to drop blocks from it, creating inconsistencies or crashing nodes.

Moreover, as there is no way to enforce the tip selection algorithm, that algorithm must be in Nash equilibrium.  If blocks are susceptible to an orphanage in mild attack scenarios or the honest setting, nodes will be highly motivated to develop new tip selection algorithms, which might jeopardise the stability of the Tangle.
 
In the first result of this article, Theorem~\ref{L-thm1},  we introduce an adversary that aims to inflate the tip pool by avoiding the selection of blocks that are tips. We then show that given the proportion $\mu$ of honest blocks being issued satisfies $\mu\leq 1/k$, where $k$ is the number of  uniformly approved blocks by each new block, then the tip pool increases over time and thus eventually bypasses the buffer capacity, yielding a successful attack, while also showing that when $\mu>1/k$, then the tip pool size has stability over time. 
 
In Theorem~\ref{theo: tip pool w/o expiration}, we investigate what a healthy tip pool looks like, both in the presence and in the absence of said attacker, showing that its size averages $\mu k\lambda h/(\mu k-1)$, while the rate of new blocks is $\lambda$, the network delay is $h$ and $\mu=1$ in the absence of the attacker. 
 
Following the result, in Theorem~\ref{theo: tip pool w/ expiration}, we propose the introduction of an expiration time $\Delta$ in the tip pool as a countermeasure against such an attack, preventing buffers from going over their capacity and giving the network a chance to recover.
We obtain that the tip pool size will average $L_0 \lambda$, where $L_0$ is the solution of the equation
 \begin{equation*}
     L_0=\frac{\mu k h}{\mu k -1 + e^{-\Delta \mu k/L_0}}.
 \end{equation*}

The aforementioned protection obtained through the expiration time introduces the possibility that blocks get removed from the tip pool without being referenced, thus becoming orphaned. 
Theorem~\ref{theo: expiration probability} shows that the probability of a block being orphaned due to expiration is bounded above by $$\exp\{-\Delta (\mu k-1)/h\},$$ which decays exponentially in $\Delta$ and is, thus, not problematic in terms of orphanage.

\subsection{Structure of the Paper}

In Section~\ref{sec: orphanage}, we discuss a specific attack that targets the tip pool of the nodes and the resulting requirement for the orphanage of blocks. In Section~\ref{sec: model}, we introduce a model for a Block DAG and the confirmation rules, which we employ to study the effects of the attack in Section~\ref{sec: theoretical results}. In Section ~\ref{sec: simulation}, we validate the model against simulation results. In Section~\ref{sec: discussion}, we provide an outlook for further research. Finally, in Section~\ref{sec: conclusion}, we discuss the obtained results.

\section{Orphanage and Tip Pool Inflation}\label{sec: orphanage}

\textit{Orphanage} of a block represents the act of a block and its content not being written in the shared version of the ledger that is eventually agreed by all nodes. There are several ways a block could be orphaned, which we will address in the following.

\subsection{Confirmation}

In blockchains, blocks are typically considered to be confirmed, i.e. are part of the ``common'' ledger, if they obtained enough support. For example, in Bitcoin, a block that is part of the longest chain and approved by a certain number of blocks is considered confirmed. Similarly, in DAG-based protocols, e.g. \cite{Gagol2019, muller2022, Prism2019}, blocks must be approved by a sufficient number of blocks or a leader block. However, different to the case of linear chains, each block can create references to multiple blocks. A block that passes a given threshold of support obtains a status \textit{confirmed} and is added to the permanent ledger.

In order to protect the node hardware from excessive memory use, limitations on the amount of time an unconfirmed block is stored must be applied. As such, a block needs to be removed from the nodes' memory if it takes excessively long to reach the confirmation status, as there is no guarantee that the block would ever reach such a status. We say a block that is dropped by the protocol by this mechanism is subject to \textit{confirmation orphanage} or \textit{liveness failure}.

We note that while the protocol should allow for confirmation orphanage to prevent storage problems, the number of orphaned blocks issued by nodes that adhere to the protocol, i.e. are issued by honest nodes, should be minimised and not be detrimental to the confirmation process.

\subsection{Expiration Orphanage}


The specific manner in which the previously mentioned limitation is achieved differs from protocol to protocol. Here we assume that each block carries a timestamp from the time it was issued. Such a timestamp may be enforced within certain bounds, e.g. in Bitcoin \cite{nakamoto2008Bitcoin}, timestamps are enforced through the order of the chain. For simplicity, we assume in our model that the clocks of all nodes are synced and that the timestamps are then monotonic with respect to the issuance order. We then enforce an age difference $\Delta$, which is the maximum time difference a block may have to any of its parent blocks. Thus, if, for a given block, any parent block has a timestamp difference of more than $\Delta$, the block is considered invalid.

It is then possible that a block does not obtain any approvers for a period $\Delta$ after it is included in the tip pool, thus being subject to \textit{expiration orphanage}.  
This orphanage can additionally lead to situations where the future cone of blocks discontinues when all tips approving it undergo expiration orphanage. This so-called \textit{future cone orphanage} will be discussed in more detail in Section~\ref{sec: discussion}. 

\subsection{Tip pool inflation attack}

Consider the following attack, called \textit{tip inflation attack}, that attempts to increase the tip pool size with the aim of increasing the  orphanage probability of blocks. 
In this attack, an adversary called \textit{spammer} issues blocks that do not approve tips. Moreover, the adversary may attach such that no contribution to confirmations is exerted (e.g. by attaching to already confirmed blocks or only to their own blocks). This attack aims to inflate the tip pool size as much as possible. This inflation has different consequences for a protocol without and with expiration time.

First, in the variant \textbf{without expiration time}, the attack targets the stability of the tip pool, i.e. increasing the tips set size over time. After a while, the increasing size of the tip pool would become a resource problem for the hardware of the node. Furthermore, as the influx of approvals is limited, timely confirmation becomes a challenge in the presence of an increasing tip pool.
In Section~\ref{sec: theoretical results}, we prove that if the spammer has a large enough share of blocks, it leads to an inflation of the tip pool. 
In \cite{ferraro2022}, such an attack was investigated without considering an expiration time. It was shown that the attack becomes a resource problem if the adversary obtains a significant issuance rate. 

For the variant \textbf{with expiration time}, the expiration time limitation acts as a protection to the tip selection process. As the tip pool size is now limited, the most the attack can achieve is to increase the confirmation orphanage rate of blocks that are issued by honest nodes.
In Section~\ref{sec: theoretical results}, we will theoretically show the impact of the tip inflation attack and the expiration orphanage in the average tip pool size, while in Section~\ref{sec: simulation}, we will validate the results obtained by comparing them to simulations.


For a practical implementation of the tip inflation attack with an expiration time, we refer to Figure~\ref{fig:practicalOrphanageAttack}. Due to the expiration time, the adversary attaches to blocks out of the tip pool. Furthermore, to maximise the orphanage of honest blocks, the adversary could also select only its own tips and, thus, never approves any blocks from honest nodes.

\begin{figure}
    \centering
    \includegraphics[width=0.45\textwidth]{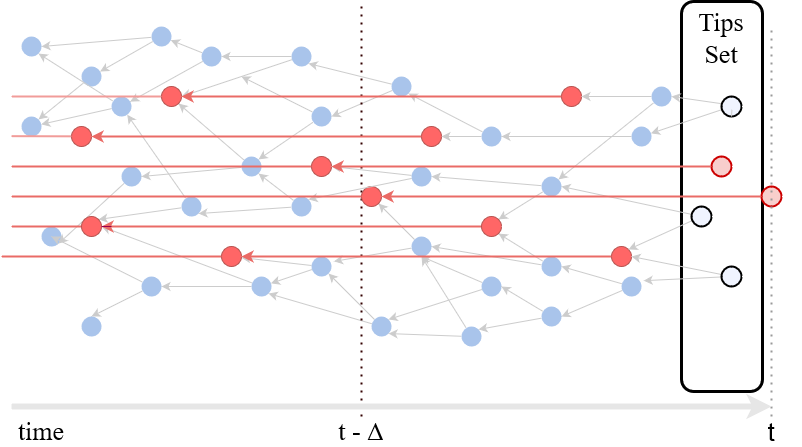}
    \caption{Practical realization of the tip inflation attack on the block DAG. The adversary  attaches to old blocks rather than to blocks that honest nodes typically find in their tip pools. }
    \label{fig:practicalOrphanageAttack}
\end{figure}



\section{Model}\label{sec: model}

Consider a set of users, also called \textit{nodes}, each with a weight (e.g.~reputation or stake) $\mathcal{N}=\{(N_1,w_1), (N_2,w_2), \ldots, (N_n,w_n)\}$. This weight is a limited resource that measures the contribution of the node in keeping the network working properly. Examples of reputation systems would be proof of work in Bitcoin and staked tokens in Proof-of-Stake networks.

Each node $N_i$ has an associated Poisson process with rate $\lambda_i = \lambda w_i/\sum_j w_j$. The parameter $\lambda$ represents the total rate of incoming blocks, while $\lambda_i$ is the proportional allocation of the resource $w_i$ that satisfies a fairness criterion by \cite{cullen2021}.

The construction of the block DAG is based on a marked Poisson process.
This is a Poisson process where each point is additionally associated with a vector of data, called \textit{mark}, of rate $\lambda$. Throughout this article, we will use the term \textit{mark} to refer to both the point (block) in the Poisson process and its associated data. Each mark $B_i$ is structured as
\begin{equation}
 \label{eq: B-def}   B_i=(t_i, \mbox{Ref}_i, \mbox{Node}_i)
\end{equation}

and ordered according to their appearance times in the Poisson process $t_i$, so $t_i<t_j$ if $i<j$. We also have $\mbox{Ref}_i$ as a subset of $k$ older marks, which will define the directed edges of our DAG, and $\mbox{Node}_i$ is the issuer node.

Consider a fixed $B_0=(-\infty,\emptyset,\emptyset)$, a fixed mark at the start of the process that we call \textit{genesis}. The view of the DAG at time $t\geq 0$ is given by $\mathcal{T}(t):= \{B_i: t_i\leq t \}$.

From the definition $\mathcal{T}(0)=\{B_0\}$. We also denote by the $\mathcal{T}$ the associated DAG built by defining directed edges between each mark and the elements of $\mbox{Ref}_i$.

\subsection{Tips and Approval}

We say that $B_i$ \textit{directly approves} $B_j$, denoted by $B_i \rightarrow B_j$, if $B_j\in \mbox{Ref}_i$. We also say that $B_i$ \textit{approves} $B_j$, and denote it by $B_i \rightsquigarrow B_j$, if there is sequence of blocks $B_{b_1}, B_{b_2},
\ldots, B_{b_M}$ such that $B_1 \rightarrow B_{b_1}$, $B_{b_m} \rightarrow B_{b_{m+1}}$ for $1\leq m \leq M-1$ and $B_{b_M} \rightarrow B_j$. This corresponds to the existence of a directed path in the DAG between the blocks $B_i$ and $B_j$. We define the set of approved blocks $\mathcal{A}(t) \subsetneq \mathcal{T}(t)$ as
\begin{align*}
    \mathcal{A}(t)&=\{B_i: B_i \in \mbox{Ref}_j \mbox{ for some $B_j$ with } t_j\leq t \}.
\end{align*}
The set of tips at time $t$, $\mbox{Tips}(t)$, is defined as $\mbox{Tips}(t)= \mathcal{T}(t)\backslash \mathcal{A}(t)$. We denote the number of tips at time $t$ by $L(t):=\vert \mbox{Tips}(t) \vert$.

The way blocks choose the $k$ elements of each set of references $\mbox{Ref}_i$ (see \eqref{eq: B-def}) defines the process. We consider the simplified model where each node, when selecting the blocks to approve, chooses only elements from the $\mbox{Tips}$ set, but with a fixed delay $h$ due to the characteristics of the network. Moreover, we assume that the selection will be done uniformly, i.e. $\mbox{Ref}_i=\mbox{Uniform}_k(\mbox{Tips}(t_i-h))$, where $\mbox{Uniform}_k$ represents the selection of $k$ independent copies of the uniform distribution with the given parameter, and $h$ represents the delay in perception that every node has regarding what is a tip due to the network delay.

To conclude the section, we remark that for the sake of succinctness, we will not present the entire probabilistic primitives, such as building the DAG space with its sigma algebra, defining the probability measure that represents the Poisson process and tip selection, and defining the filtration to make $\mathcal{T}(t)$ a proper stochastic process. Although of technical importance, those elements can be inferred from our definitions and play no role in the proofs to come. 

\section{Theoretical Result}\label{sec: theoretical results}

This section presents theoretical results on the tip inflation attack, its effects on the tip pool, and the expiration orphanage. We start by giving an important definition that mathematically represents the long-term stability of the process: 
\begin{defi}[Asymptotic Stationarity] We say a process $X_t$ is asymptotically stationary if there exists a random variable $X$ such that $X_t \to X$ almost surely as $t\to \infty$.
\end{defi}

\subsection{Stationarity}

According to the DAG model we defined, tips may be classified into three different categories at any point $t$ in time according to their tip status and the local view of the DAG. For this purpose, recall that if the DAG is at time $t$, each node has a local view of the DAG at time $t-h$. With this, tips may be classified as:
\begin{enumerate}
    \item \textbf{Hidden Tips.} Blocks that were issued in the interval $(t-h;t]$ and thus are still not known to other nodes due to the network delay.
    \item \textbf{Real Tips.} These are blocks that were issued before time $t-h$ and did not receive any approval from their issuance time until time $t$. The number of real tips at time $t$ is denoted by $R(t)$.
    \item \textbf{False Tips.} Finally, those are blocks that were issued before time $t-h$, did not receive any approval until time $t-h$ but received at least one approval in the interval $(t-h;t]$. Since the block that approved them is not visible to other nodes due to network delay, this block is considered a tip by nodes even though it is already referenced. The number of false tips at time $t$ is denoted by $F(t)$.
\end{enumerate}

We will show under which conditions the tip pool size will have  an increasing  or reducing tendency (positive or negative drift). In situations with negative drift, then by supermartingale arguments, we can conclude the convergence to a proper random variable, which is exactly the asymptotic stationarity we are looking for. 

Consider the discrete-time process $X_n$, defined as the tip pool sizes after the arrival of each block: $X_n=L(t_n)$. The long-term behaviour in $n$ of $X_n$ is equivalent to the one  in $t$ from $L(t)$. 

\begin{theo}\label{L-thm1} For any $x>0$ it holds almost surely that
\begin{equation}
    \label{eq: Xnlower}\ex[X_{n}-X_{n-1}\vert X_{n-1}=x] > 1- \mu k,
\end{equation}
and thus if $\mu k\leq 1$ the process of tips increases unbounded as $t\to\infty$, diverging to infinite. Moreover let $\mu k> 1$ and fix $0<\varepsilon'<\mu k -1$. Then there exists $x'$ depending on $\varepsilon'$ such that if $x>x'$ it holds almost surely that
\begin{equation}
   \label{eq: Xnupper}\ex[X_{n}-X_{n-1}\vert X_{n-1}=x] \leq -\varepsilon'.
\end{equation}
Thus if $\mu k> 1$, the tip pool size has negative drift for large values and is asymptotically stationary as $t\to \infty$. 
\end{theo}

\begin{proof}
 The difference $X_{n}-X_{n-1}$ represents the changes that the block $B_{n}$ makes in the tip pool size; this block will add himself as a (hidden) tip, and in the process remove a certain quantity of tips. The amount removed is associated with how many fake tips and how many real tips were selected by his selection algorithm, and if the issuer was honest or the spammer (in such cases, $0$ tips are removed). Calling $D_n=X_n-X_{n-1}$ the variable representing this change, it's then straightforward that
 \begin{equation}\label{eq-Dn1}
     \ex[D_n\vert X_{n-1}=x]=1-\mu \ex[SR_{n}\vert X_{x-1}=x],
 \end{equation}
where $SR_n$ is the number of real tips selected by $B_{n}$. We start by proving \eqref{eq: Xnlower}, observe that at most all distinct blocks chosen in the tip selection  for $B_n$ will be real tips, and hence $SR_{n} \leq DB_{n}$, where $DB_n$ is the number of distinct blocks chosen by $B_n$. Now observe that by definition $DB_n\leq k$ and that $\pr[DB_n=k\vert X_{n-1}=x]<1$ since independently of the tip pool size at time $t-h$, the selection of $B_n$ is independent and uniform. With this 
\begin{equation}
    \label{eq-DB}\ex[SR_{n}\vert X_{x-1}=x] \leq \ex[DB_n\vert X_{n-1}=x]< k.   
\end{equation}
Applying \eqref{eq-DB} in \eqref{eq-Dn1} yields \eqref{eq: Xnlower}.

Now onto \eqref{eq: Xnupper}: When the selection happens, $B_{n}$ make a selection from $L(t_n-h)$, but these are split between $R(t_n)$ and $F(t_n)$, what essentially make $SR_{n}$ the number of distinct elements chosen among $R(t_n)$, and the number of selections is a conditional binomial with $k$ selections and event probability $R(t_n)/L(t_n-h)$. Applying this to \eqref{eq-Dn1} we get
\begin{equation} \label{eqD2}
    \ex[D_n\vert X_n=x]\leq 1-\mu k \ex[R(t_n)/L(t_n-h)\vert X_n=x].
\end{equation}
Let $K_{t_n}=| \{i: t_i \in (t_n-h,t_n]\}|$ be the number of blocks issued at most $h$ before $t_n$.
Observe now that $L(t_n-h)\geq L(t_n)-k N_{t_n}$ and $R(t_n)\geq L(t_n-h)-k K_{t_n}$, where the second result comes from the fact that in the worst case scenario each new issued block in the interval $(t_n-h,t]$ would remove $k$ tips from the ones in $L(t_n-h)$. Using this, \eqref{eqD2} becomes  
\begin{align}
\label{eqD3}  \ex[D_n\vert X_n=x]&\leq 1-\mu k \ex\Big[\frac{L(t_n)-2k K_{t_n}}{L(t_n)-kK_{t_n}}\vert X_n=x\Big]\\
&=1-\mu k \ex\Big[\frac{x-2kK_{t_n}}{x-kK_{t_n}}\vert X_n=x\Big].
\end{align}
Now, since $K_{t_n}$ is a Poisson variable with exponential tails, we can choose a constant $C_1$ such that $\pr[K_{t_n}<C_1]>(1-\varepsilon)$ and so
\begin{equation}
    \label{eqD3b}  \ex[D_n\vert X_n=x]\leq 1- \mu k(1-\varepsilon) \frac{x-2k C_1}{x-k C_1}.
\end{equation}
So if $\mu >1/k$, we can chose $\varepsilon$ small enough (and thus $C_1$ large enough) such that there is a $x'$ large enough that if $x>x'$, then $\ex[D_n\vert X_n=x]\leq -\varepsilon'<0$. This guarantees a negative drift, and, together with standard aperiodicity considerations, the claim follows. 
\end{proof}

\begin{rem}
Using the asymptotic negative drift to prove stationarity is inspired by \cite{kumar2022effect} and \cite{Mu:stability}, where the stationarity of the local tip pool sizes is studied. 
\end{rem}


\subsection{Number of Tips}

We present the following proposition

\begin{prop}\label{L-prop2} Let $\mu k>1$ and consider that, from Theorem~\ref{L-thm1}, $L(t)\to_{t\to\infty}L_{\infty}$, then $L_{\infty}$ satisfies $L_{\infty}/\lambda \to L_0$ almost surely as $\lambda \to \infty$, where $L_0$ is a constant.
\end{prop}

This is a consequence of Theorem~\ref{L-thm1} and the interpretation of the tip pool as a queuing process, where tips represent customers in a system. A more detailed description of such a relation with further theoretical backing can be found in \cite{equilibria}.

\begin{theo}\label{theo: tip pool w/o expiration}[Tip pool size without expiration] If the rate of incoming blocks is all honest, i.e. $\lambda_H = \lambda$, then the number of tips satisfies almost surely as $\lambda\to \infty$
\begin{equation}\label{eqn:L_noAdversary}
L(t)\to \frac{k}{k-1}\lambda h.
\end{equation}
Moreover, under the attack of a spammer with rate $\lambda_S$ such that $\lambda=\lambda_H+\lambda_S$ and fixed $\mu:=\lambda_H/\lambda$, if  $\mu k>1$ the number of tips satisfies almost surely as $\lambda\to \infty$
\begin{equation}\label{eqn:L_ModelA}
L(t)\to \frac{\mu k}{\mu k-1}\lambda h.
\end{equation}
\end{theo}

 This result is already known for the scenario with the absence of spammers from \cite{thetangle}, \cite{properties}, and \cite{cullen2019-variable-delay}. While the technique used in \cite{thetangle} and \cite{properties} can be applied to our setting, we propose another technique that turns out  useful for the  proof of Theorem \ref{theo: tip pool w/ expiration}.

\begin{proof}
For each block $B_j$, denote the time it spends as a tip by $\tau_j$, i.e. the time until another block approves it. The main idea of this proof is to try to  apply the so-called Little's law from queuing theory \cite{little} to the tip selection problem, summing the tip time of blocks that existed and comparing the result with the average number of tips. Let us define the total tip time of the block DAG up to time $t$:
\[
\Theta(t) := \sum \limits_{B_i \in \mathcal{T}(t)} \min\{ \tau_i,t-t_i \}.
\]

The minimum term in the summand is not to let the tip time go over the current time (if $t_i+\tau_i>t$).

Since a tip is exactly a block whose tip time has not ended, we also have
\[
\Theta(t)=\int \limits_{0}^tL(s) ds.
\]
Now for any element $\omega$ of the sample space such that $L(t):= L(t,\omega) \to L_\infty (\omega)$, we have 
\begin{equation}
    \label{eq: int-diff}\frac{1}{t}\int \limits_{0}^t\vert L(s,\omega)-L_\infty (\omega) \vert ds \to 0
\end{equation}

This is straightforward since $\vert L(s,\omega)-L_\infty (\omega) \vert<\varepsilon$ implies that the left-hand side of \eqref{eq: int-diff} is also smaller than $\varepsilon$. Since from Theorem \ref{L-thm1} the sample path where \eqref{eq: int-diff} holds measure $1$, the almost sure limit
\begin{equation}
    \frac{1}{t}\int \limits_{0}^t L(s) ds \to L_\infty, \mbox{ as }t\to \infty,
\end{equation}
 holds. Finally, using Proposition \ref{L-prop2}, it holds almost surely that
 \begin{equation}
     \lim_{\lambda\to \infty} \lim_{t\to \infty} \frac{\Theta (t)}{t\lambda} = L_0.
 \end{equation}

Under the asymptotic stationarity of the number of tips, we have that $\tau_i$ are asymptotically identically distributed and with constant variance. Hence we can apply a strong law of large numbers to obtain
\begin{align}
        \nonumber\lim_{\lambda\to \infty}\lim_{t\to \infty}\frac{\Theta(t)}{\lambda t}&=\lim_{\lambda\to \infty}\lim_{t\to \infty}\frac{1}{\lambda t}\sum \limits_{B_i \in \mathcal{T}(t)} \min\{ \tau_i,t-t_i \}\\
       \label{eq: theta2}&= \lim_{\lambda\to \infty}\lim_{i\to \infty} \ex \tau_i.
\end{align}
From what we can conclude
\begin{equation}\label{eq: theta3}
    L_0=\lim_{\lambda\to \infty} \lim_{i\to \infty} \ex \tau_i.
\end{equation}

This concludes Little's law analogue to the DAG block issuance process. Observe that the limit on $i$ in \eqref{eq: theta3} is equivalent to a limit in $t$. 

We now continue to calculate $\ex \tau_i$. After the issuance of $B_i$, it stays as a hidden tip for time $h$, and then every new block in the process has a probability of selecting it (and thus finishing its time as a tip)
\begin{equation}\label{eq: pLt}
    p(L(t))=1-(1-1/L(t))^k,
\end{equation}
where $t$ is the issuance time for the block. This means that the processes of approvals of $B_i$, until time $h$ after the first approval, follows a non-homogeneous Poisson process $P_{B}(t)$ with rate function $\lambda p(L(t))$. The expectation of the first arrival in such a process (representing the quantity $\tau_i-h$ for us) is known, and by \eqref{eq: theta3}, we have
\begin{align}
  \nonumber L_0&-h =\lim_{\lambda\to \infty}\lim_{i\to \infty} \ex \tau_i-h \\
  \nonumber &=\lim_{\lambda\to \infty}\lim_{i\to \infty} \ex \int \limits_{0}^\infty x\lambda p(L(t_i+x))e^{-\int_0^x \lambda p(L(t_i+y))dy}dx  \\
  \label{eq: L-limit0}&=\lim_{\lambda\to \infty}\ex \int \limits_{0}^\infty x\lambda p(L_\infty)e^{-\int_0^x \lambda p(L_\infty)dy}dx\\
  \nonumber &= \lim_{\lambda\to \infty}\ex \int \limits_{0}^\infty x\lambda p(L_\infty)e^{-x\lambda p(L_\infty)}dx=\ex  \lim_{\lambda\to \infty} \frac{1}{\lambda p(L_\infty)}\\
  \label{eq: L-limit}&=\ex  \lim_{\lambda\to \infty}\frac{1}{\lambda \Big(1-(1-\frac{1}{L_\infty})^k \Big)}=\frac{L_0}{k},
\end{align}
from where we obtain $L_0 = k h/(k-1)$, concluding the proof in the honest scenario.

For the tip pool inflation scenario, consider a node (or nodes) $N_S$ performing the tip inflation attack with issuance rate $\lambda_S$, while the rest of the network is honest with rate $\lambda_H=\sum_{i\neq S}^{n}\lambda_i$, forming a total rate $\lambda=\lambda_S+\lambda_H$.

Observe that the average tip time for each tip does not change independently of the tip being from the spammer or honest nodes since the spammer ignores all tips when selecting a tip, and the honest nodes do not distinguish between them. This means that the calculations that led to \eqref{eq: theta3} still hold in this scenario. In contrast to the honest scenario, to calculate the expected tip time, here we need to consider that in the Poisson process $P_{B}(t)$ of approvals of a block, only the honest rate $\lambda_H$ has any chance of approving tips, changing then the rate of $P_{B}(t)$ to $\lambda_H p(L(t))=\mu \lambda p(L(t))$. Replicating the calculations that led to \eqref{eq: L-limit} with the new rate, we obtain 
\begin{align}
  \label{eq: L-limit2} L_0-h &=\ex  \lim_{\lambda\to \infty}\frac{1}{\mu\lambda p(L_\infty)}=\frac{L_0}{\mu k}
\end{align}
solving the equation yields $L_0 = h\mu k /(\mu k -1)$, which concludes the second part of the proof. 

\end{proof}

\begin{theo}\label{theo: tip pool w/ expiration}[Tip pool size with expiration]
Consider the same setup as in Theorem~\ref{theo: tip pool w/o expiration} with spammers performing a tip inflation attack. Moreover, let tips that stay in the tip pool for a period longer than $\Delta$ be dropped (i.e. expire). Under these conditions, the  tip pool size $L(t)$ behaves asymptotically, for large $t$ and large $\lambda$, as $L_0 \lambda$, where $L_0$ is the solution of the equation 
\begin{equation}\label{eq: L0}
    L_0=\frac{\mu k h}{\mu k -1 + e^{-\Delta \mu k/L_0}}.
\end{equation}
\end{theo}

\begin{proof}
In this scenario, we also assume the existence of nodes with the total rate $\lambda_S$ performing a tip inflation attack, while the rest of the network is honest with the rate  $\lambda_H=\lambda - \lambda_S$. The main difference in this scenario is the introduction of a limit $\Delta$ on the time that a block can stay as a tip before being dropped from the pool and, thus, being orphaned. 

 In this scenario, the asymptotic relation from \eqref{eq: theta3} still holds, but the tip time for each block is now bounded by $\Delta$. If we denote by $E_i$ a random variable representing the first arrival of the Poisson process with rate $r(t)=\lambda_H p(L(t_i+t))$, we have
\begin{equation}\label{eq: tauDelta}
    \tau_i \laweq h+\min\{E_i, \Delta\},
\end{equation}

Using \eqref{eq: tauDelta} and the asymptotic relation \eqref{eq: theta3} we get
\begin{align}
\nonumber L_0-h &= \lim_{\lambda\to \infty}\lim_{i\to \infty}\ex \tau_i-h\\
\nonumber&= \lim_{\lambda\to \infty}\lim_{i\to \infty}\ex\int \limits_{0}^{\Delta} x r(x) e^{-\int_0^x r(s) ds}dx\\
\label{eq: NHPoisson1}& \quad\quad +\Delta \lim_{\lambda\to \infty}\lim_{i\to \infty} \pr[E_i\geq \Delta].
\end{align}

From Proposition~\ref{L-prop2}, for every $x$ it holds $r(x)\to \lambda_H p(L_\infty)$  almost surely as $i \to \infty$, so we can use the same calculation as in \eqref{eq: L-limit0} for \eqref{eq: NHPoisson1} to get

\begin{align}
\nonumber L_0-h &= \lim_{\lambda\to \infty}\ex\int \limits_{0}^{\Delta} x \lambda_H p(L_\infty) e^{-x\lambda_H p(L_\infty)}dx\\
\nonumber & \quad\quad +\Delta \lim_{\lambda\to \infty} \ex e^{-\Delta\lambda_H p(L_\infty)}\\
\nonumber &=\ex\lim_{\lambda\to \infty}\frac{1-(\Delta\lambda_H p(L_\infty)+1)e^{-\Delta\lambda_H p(L_\infty)}}{\lambda_H p(L_\infty)}\\
\label{eq: NHPoisson2} & \quad +\Delta \ex\lim_{\lambda\to \infty} e^{-\Delta\lambda_H p(L_\infty)}.
\end{align}

From \eqref{eq: L-limit2} we know that $\lambda_H p(L_\infty)\to \mu k/L_0$ almost surely as $\lambda \to \infty$. Using this in \eqref{eq: NHPoisson2} yields
\begin{align}
\label{eq: tauDelta2} L_0-h&=\frac{L_0-(\Delta \mu k+L_0)e^{-\frac{\Delta \mu k}{L_0}}}{\mu k}+\Delta e^{-\frac{\Delta \mu k}{L_0}}.
\end{align}
Analysing the equation above as $L_0\to 0$ and $L_0\to \infty$, we can be sure that a solution always exist. Rearranging the terms of \eqref{eq: tauDelta2} finishes the proof of the theorem.

\end{proof}


\subsection{Expiration Orphanage Probability}

Now that we have a good understanding of the tip pool size, we can calculate the probability that a block gets orphaned in the tip pool due to expiration. 

\begin{theo}\label{theo: expiration probability}[Expiration Probability] Consider the same setup as in Theorem~\ref{theo: tip pool w/ expiration}. Then the probability that a block is dropped from the tip pool (i.e. expires) for staying in it over $\Delta$ units of time satisfies
\begin{align}
    \nonumber \lim \limits_{\lambda\to \infty}\lim \limits_{j\to \infty}&\pr[\text{Block }B_j\text{ expires}]= \exp \Big\{- \frac{\Delta \mu k}{L_0 }\Big\}\\
\label{expire-Ineq}&\leq\exp \Big\{- \frac{\Delta (\mu k-1)}{h }\Big\}.
\end{align}
\
\end{theo}

The limit on $j$ represents the asymptotic stationarity of the process. Simulations from \cite{ferraro2022} show that the convergence of the process to its stationary measure is typically quick for a reasonable range of spammer proportion, and so the asymptotic results we obtain are close to what one expects in a realistic scenario. Now onto the proof of the theorem.
\begin{proof}
Recall that the process of approvals that a fixed block $B_j$ receives is a conditional non homogeneous Poisson process $P_j(t)$ with rate $\lambda_H p(L(t))$, hence the expiration event $E_j$ for block $B_j$ satisfies
\begin{equation}
    \lim_{j\to \infty} \pr[E_j]= \lim_{j\to \infty}  \pr[P_j(t_j+h,t_j+h+\Delta]=0].
\end{equation}
But $P_j(t_j+h,t_j+h+\Delta]$ is a Poisson random variable with parameter $\int_{0}^{\Delta} \lambda_H p(L(t_j+h+s))ds$ and thus by Proposition~\ref{L-prop2} we have
\begin{align*}    
    \lim_{\lambda\to \infty} \lim_{j\to \infty} \pr[E_j]&=\lim_{\lambda\to \infty}\lim_{j\to \infty}\ex e^{-\int_{0}^{\Delta} \lambda_H p(L(t_j+h+s))ds}\\
    &=\ex \lim_{\lambda\to \infty}\exp \Big\{-\Delta \lambda_H p(L_\infty)  \Big\}\\
    &=\exp\Big\{-\Delta \mu h/L_0\Big\}.
\end{align*}
The last relation comes from continuity and \eqref{eq: L-limit2}.

The inequality \eqref{expire-Ineq} is achieved by replacing $L_0$ by the value in Theorem \ref{theo: tip pool w/ expiration}, and using that
\[\mu k -1 \leq \mu k -1 +  \exp \Big\{- \frac{\Delta \mu k}{L_0 }\Big\}.\]
\end{proof}

\section{Simulation Studies and Validation}\label{sec: simulation}

We compare the derived models against studies using the open-sourced simulation code \cite{tangleSim}. 
We make the following assumptions in the simulator, which have previously also been utilised in, e.g.~\cite{Kusmierz2019,thetangle}. Blocks are created through a Poisson process, and we maintain a global view of the block DAG and the tip pool. To emulate the propagation of new blocks in the network, blocks are added to the global tip pool with a delay of $h$ after their creation. From the global tip pool, parent tips are selected during the tip selection mechanism. For each created block, we choose the honest tip selection with probability $\mu$ or the adversary tip selection otherwise. Blocks for which a referencing block is added to the tip pool or that are more than $\Delta$ old are removed from the tip pool. To obtain high precision, for each data point, we run the simulation $n=100$ times and where each run yields a block DAG with at least $300.000$ blocks. For the tip pool sizes, this yields an error (standard deviation) on the mean value of less than 1\%. For the orphanage rate, we indicate the standard deviation of the results in the figure. We study parameters, as shown in Table \ref{tab:default simulation values}.


\begin{figure}
    \centering
    \includegraphics[width=0.45\textwidth]{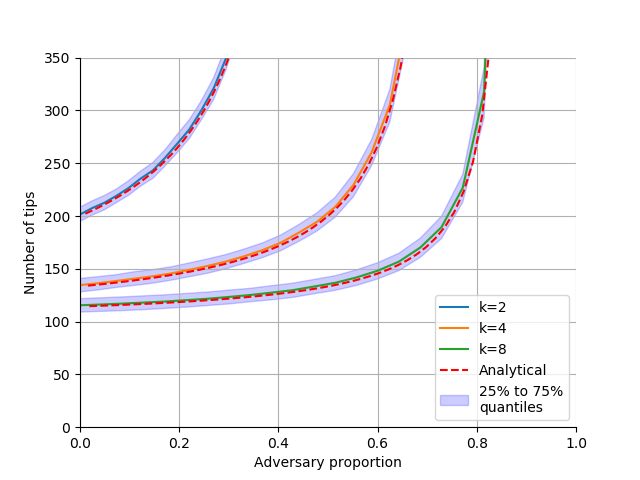}
    \caption{Tip pool sizes with the proportion of blocks issued by the  adversary. The mean values of simulations are compared with the analytical value for several $k$. The variance of the tip pool is shown through quantile ranges. }
    \label{fig:tipPool}
\end{figure}

\begin{table}[h]
    \centering
    \caption{Default parameters}
    \begin{tabular}{l|c}
         Parameter & Value  \\
         $\lambda$ & 100 blocks / $h$ \\
         $\Delta $ & 100 $h$\\
         $\mu$ & $\{0,..,1.\}$\\
         $k$ & $\{1,..,8\}$
    \end{tabular}
    \label{tab:default simulation values}
\end{table}

First, to illustrate the challenge of high adversary issuance rates, we study in Figure~\ref{fig:tipPool} the tip pool sizes considering no restriction on the expiration time, i.e.~$\Delta=\infty$. The analytical values agree well with the mean values of the simulation results. We also show the variance of the tip pool by providing the 25\% and  75\% quantiles. As by Theorem~\ref{L-thm1} the tip pool is no longer stationary for $\mu k>1$, which can be seen through the diverging tip pool size for  $\mu$ close to those values. We can see that by increasing the number of parent references, we can substantially increase the robustness of the protocol against this type of attack.

We counter the divergence by setting $\Delta$ to a finite value. In Figure~\ref{fig:tipPool_Delta}, we show the tip pool sizes with the proportion of blocks issued by the adversary for $\Delta=100h$.
Due to the expiration, the tip pool size converges, and the obtained values agree with the model. Clearly, with the introduction of the expiration, the tip pool size is limited by the number of blocks issued within $\Delta$, even if the issuance of blocks is dominated by the adversary.

\begin{figure}[t]
    \centering
    \includegraphics[width=0.45\textwidth]{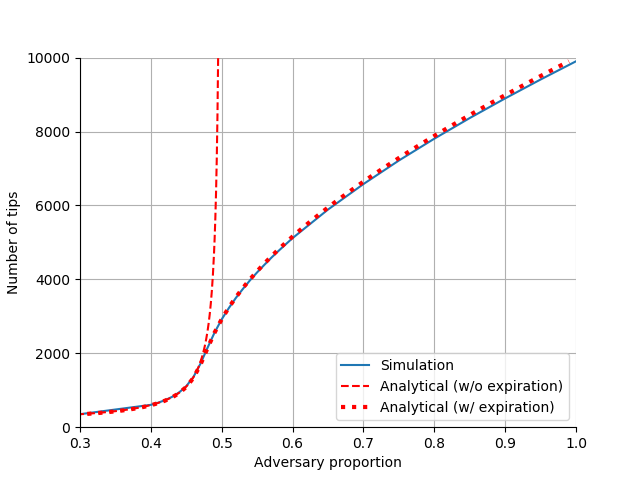}
    \caption{Tip pool sizes with the proportion of blocks issued by the adversary, for $\Delta=100h$. }
    \label{fig:tipPool_Delta}
\end{figure}

Figure ~\ref{fig:orphanage_Delta} shows the orphanage rates with the proportion of blocks issued by the adversary. For the simulation results, the standard deviation of the error is displayed, which is low for all data points due to a large number of samples. 
As can be seen, the analytical model ($L_0$ with expiration) agrees well with the simulation-obtained values. 
For $\mu>1/k$, the probability for orphanage reduces drastically and the rate of the exponential decay matches for all curves.

Lastly, we also consider the values for $L_0$ obtained from Theorem~\ref{theo: tip pool w/o expiration} (without expiration) as input to Theorem~\ref{theo: expiration probability} to support the later discussion in Section~\ref{sec: discussion} on future cone orphanage.


\begin{figure}[t]
    \centering
    \includegraphics[width=0.45\textwidth]{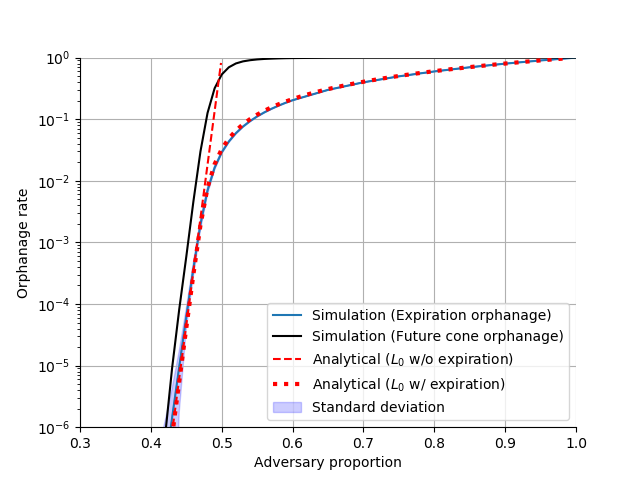}
    \caption{Orphanage probability with the proportion of blocks issued by the adversary, for $\Delta=100h$, $k=2$. }
    \label{fig:orphanage_Delta}
\end{figure}

\section{Discussion}\label{sec: discussion}

In this section, we provide an outlook on the topic of \textit{future cone orphanage}, which extends the concept of expiration orphanage. 


As mentioned earlier, it is possible for a block that has obtained one or more approvers that all tips of its future cone become expiration orphaned. Therefore, the DAG discontinues to grow on top of such a block. If the block does not reach a sufficient number of blocks referencing it by the time the tips of its future cone are orphaned, it cannot reach a confirmation status itself. A block for which the future cone is discontinued in such a way is said to be subject to future cone orphanage. 

Future cone orphanages can also lead to agreement failures. Consider, for example, the process of a new node joining the network. Such a node may learn about new blocks, which are tips, through the connection to its peers. It then recursively can request missing blocks in the past cone of these new blocks through a process called solidification. As a consequence, only tips that are reachable from the current set of gossiped blocks are known to the joining node. However, it is possible that some blocks may get confirmed despite becoming future cone orphaned, leading to agreement failure between the newly joined node and the rest of the network. 

We study the orphanage probability in simulation studies, similar to the previous section. In
Figure ~\ref{fig:orphanage_Delta} we consider only blocks that are at least $\tau=3\Delta$ in the past of the last tips. We remove more recent blocks from the analysis since they do not represent the future cone orphanage probability well. More specifically, blocks that only recently were removed from the tip pool are more likely to still have tips in their future cone that are not yet expired. As such, the calculated value of the future cone orphanage expresses a lower bound since even from the considered blocks a proportion may become future cone orphaned, eventually. As can be seen from the figure, if we consider the value for $L_0$ calculated by Theorem~\ref{theo: tip pool w/o expiration} in Theorem~\ref{theo: expiration probability}, i.e. ``w/o expiration'', we are provided with  better estimates for the future cone orphanage. 


\section{Conclusion}\label{sec: conclusion}

In DAG-based DLTs, blocks can be added asynchronously and concurrently by referencing previous blocks. To preserve good performances of the resulting block DAG, the stability of the tip pool is important. We presented an attack that attempts to inflate the tip pool and consequently results in resource depletion of the nodes. To protect the node, we propose the introduction of an expiration time on the tips. 
We present models to analytically predict the tip pool sizes in the scenarios with and without an expiration time. We prove that, without expiration, the attack manages to break the stationarity (stability) of the tip pool if the adversary obtains sufficient writing access. On the other hand, by  introducing expiration times, the tip pool size remains limited, which may allow nodes to recover if the attack is finite in time. Furthermore, we provide theoretical results for the expiration orphanage rates and show that this rate reduces exponentially if the adversary issues blocks at a rate less than some critical value.

We compare the analytical results with simulation results using a simulation tool. Since the expiration orphanage can be extended to the future cone orphanage, we  utilise the simulator also to study the future cone orphanage.

\bibliographystyle{IEEEtran}
\bibliography{bibliography}

\begin{thebibliography}{10}
\providecommand{\url}[1]{#1}
\csname url@samestyle\endcsname
\providecommand{\newblock}{\relax}
\providecommand{\bibinfo}[2]{#2}
\providecommand{\BIBentrySTDinterwordspacing}{\spaceskip=0pt\relax}
\providecommand{\BIBentryALTinterwordstretchfactor}{4}
\providecommand{\BIBentryALTinterwordspacing}{\spaceskip=\fontdimen2\font plus
\BIBentryALTinterwordstretchfactor\fontdimen3\font minus
  \fontdimen4\font\relax}
\providecommand{\BIBforeignlanguage}[2]{{%
\expandafter\ifx\csname l@#1\endcsname\relax
\typeout{** WARNING: IEEEtran.bst: No hyphenation pattern has been}%
\typeout{** loaded for the language `#1'. Using the pattern for}%
\typeout{** the default language instead.}%
\else
\language=\csname l@#1\endcsname
\fi
#2}}
\providecommand{\BIBdecl}{\relax}
\BIBdecl

\bibitem{Gagol2019}
A.~Gagol, D.~Le{\'s}niak, D.~Straszak, and M.~{\'S}wietek, ``Aleph:
  Efficient atomic broadcast in asynchronous networks with byzantine nodes,''
  in \emph{Proceedings of the 1st ACM Conference on Advances in Financial
  Technologies}, 2019, pp. 214--228.

\bibitem{sui2022}
\BIBentryALTinterwordspacing
MystenLabs, ``{The Sui Smart Contracts Platform },'' 2022. [Online]. Available:
  \url{https://github.com/MystenLabs/sui/blob/main/doc/paper/sui.pdf}
\BIBentrySTDinterwordspacing

\bibitem{Danezis2022}
\BIBentryALTinterwordspacing
G.~Danezis, L.~Kokoris-Kogias, A.~Sonnino, and A.~Spiegelman, ``{Narwhal and
  Tusk: A DAG-Based Mempool and Efficient BFT Consensus},'' in
  \emph{Proceedings of the Seventeenth European Conference on Computer
  Systems}, ser. EuroSys '22.\hskip 1em plus 0.5em minus 0.4em\relax New York,
  NY, USA: Association for Computing Machinery, 2022, p. 34–50. [Online].
  Available: \url{https://doi.org/10.1145/3492321.3519594}
\BIBentrySTDinterwordspacing

\bibitem{muller2022}
S.~Müller, A.~Penzkofer, N.~Polyanskii, J.~Theis, W.~Sanders, and H.~Moog,
  ``{Tangle 2.0 Leaderless Nakamoto Consensus on the Heaviest DAG},''
  \emph{IEEE Access}, 2022.

\bibitem{Avalanche19}
T.~Rocket, M.~Yin, K.~Sekniqi, R.~van Renesse, and E.~G. Sirer, ``Scalable and
  probabilistic leaderless {BFT} consensus through metastability,'' 2019.

\bibitem{CongestionControl}
A.~Cullen, P.~Ferraro, W.~Sanders, L.~Vigneri, and R.~Shorten, ``{On Congestion
  Control for Distributed Ledgers in Adversarial IoT Networks},'' \emph{CoRR},
  vol. abs/2005.07778, 2020.

\bibitem{Prism2019}
V.~Bagaria, S.~Kannan, D.~Tse, G.~Fanti, and P.~Viswanath, ``{Prism:
  Deconstructing the Blockchain to Approach Physical Limits},'' in
  \emph{Proceedings of the 2019 ACM SIGSAC Conference on Computer and
  Communications Security}, ser. CCS '19.\hskip 1em plus 0.5em minus
  0.4em\relax New York, NY, USA: Association for Computing Machinery, 2019, p.
  585–602.

\bibitem{nakamoto2008Bitcoin}
S.~Nakamoto, ``{Bitcoin: A peer-to-peer electronic cash system},'' 2008.

\bibitem{ferraro2022}
\BIBentryALTinterwordspacing
P.~Ferraro, A.~Penzkofer, C.~King, and R.~Shorten, ``Feedback control for
  distributed ledgers: An attack mitigation policy for dag-based dlts,'' 2022.
  [Online]. Available: \url{https://arxiv.org/abs/2204.11691}
\BIBentrySTDinterwordspacing

\bibitem{cullen2021}
A.~Cullen, P.~Ferraro, W.~Sanders, L.~Vigneri, and R.~Shorten, ``{Access
  Control for Distributed Ledgers in the Internet of Things: A Networking
  Approach},'' \emph{IEEE Internet of Things Journal}, pp. 1--1, 2021.

\bibitem{kumar2022effect}
N.~Kumar, A.~Reiffers-Masson, I.~Amigo, and S.~Ruano~Rinc{\'o}n, ``The effect
  of network delays on distributed ledgers based on direct acyclic graphs: A
  mathematical model,'' \emph{Available at SSRN 4253421}, 2022.

\bibitem{Mu:stability}
\BIBentryALTinterwordspacing
S.~Müller, I.~Amigo, A.~Reiffers-Masson, and S.~Ruano-Rincón, ``Stability of
  local tip pool sizes,'' 2023. [Online]. Available:
  \url{https://arxiv.org/abs/2302.01625}
\BIBentrySTDinterwordspacing

\bibitem{equilibria}
S.~Popov, O.~Saa, and P.~Finardi, ``Equilibria in the tangle,'' \emph{Computers
  \& Industrial Engineering}, vol. 136, pp. 160--172, 2019.

\bibitem{thetangle}
S.~Popov, ``{The Tangle},'' \emph{Version 1.4.3}, 2018.

\bibitem{properties}
B.~Kusmierz, W.~Sanders, A.~Penzkofer, A.~Capossele, and A.~Gal, ``Properties
  of the tangle for uniform random and random walk tip selection,'' in
  \emph{2019 IEEE International Conference on Blockchain (Blockchain)}.\hskip
  1em plus 0.5em minus 0.4em\relax IEEE, 2019, pp. 228--236.

\bibitem{cullen2019-variable-delay}
A.~Cullen, P.~Ferraro, C.~King, and R.~Shorten, ``Distributed ledger technology
  for smart mobility: Variable delay models,'' in \emph{2019 IEEE 58th
  Conference on Decision and Control (CDC)}, 2019, pp. 8447--8452.

\bibitem{little}
J.~D. Little, ``A proof for the queuing formula: {$L= \lambda W$},''
  \emph{Operations research}, vol.~9, no.~3, pp. 383--387, 1961.

\bibitem{tangleSim}
``{Simulator: Res-Sim-Tangle},'' 2022,
  \url{https://github.com/iotaledger/res-sim-tangle}.

\bibitem{Kusmierz2019}
B.~Kusmierz, W.~Sanders, A.~Penzkofer, A.~Capossele, and A.~Gal, ``Properties
  of the tangle for uniform random and random walk tip selection,'' in
  \emph{2019 IEEE International Conference on Blockchain (Blockchain)}, 2019,
  pp. 228--236.

\end{thebibliography}


\end{document}